\documentclass[12pt]{article} 

\usepackage{dsfont}
\usepackage{setspace}

\usepackage{amsmath, amsthm, amssymb, amsfonts, enumerate}

\usepackage[longnamesfirst]{natbib}

\newcommand{\conv}{\operatorname{conv}}

\newtheorem{theorem}{Theorem}

\newtheorem{lemma}[theorem]{Lemma}

\theoremstyle{definition}

\newtheorem{definition}[theorem]{Definition}

\newtheorem{example}[theorem]{Example}

\newtheorem{remark}[theorem]{Remark}

\numberwithin{equation}{section}

\numberwithin{theorem}{section}

\newcommand{\ignore}[1]{}

\begin{document}

\title{On the Core of Dynamic Cooperative Games}
\author{Ehud Lehrer%
\thanks{This author
acknowledges the support of the Israel Science Foundation, Grant
\#762/045.}\\
              School of Mathematical Sciences \\
              Tel Aviv University \\
              Tel Aviv 69978, Israel \\
              and 
              INSEAD \\
              Bd. de Constance \\
              F--77305 Fontainebleau Cedex, France.  \\
              \texttt{lehrer@post.tau.ac.il}  
\and 
Marco Scarsini \\
         Dipartimento di Economia e Finanza \\
         LUISS \\
	Viale Romania 12 \\
	I--00197 Roma, Italy. \\
	 \texttt{marco.scarsini@luiss.it}
}

\date{\today}

\maketitle

\thispagestyle{empty}

\pagebreak
 
\begin{abstract}
We consider dynamic cooperative games, where the worth of coalitions
varies over time according to the history of allocations. When
defining the core of a dynamic game, we allow the possibility for
coalitions to deviate at any time and thereby to give rise to a new
environment. A coalition that considers a deviation needs to take the 
consequences into account because from the deviation point on, the
game is no longer played with the original set of players. The
deviating coalition becomes the new grand coalition which, in turn,
induces a new dynamic game. The stage games of the new dynamical
game depend on all previous allocation including those that have
materialized from the deviating time on.

We define three types of core solutions: fair core, stable core and
credible core. We characterize the first two in case where the
instantaneous game depends on the last allocation (rather than on
the whole history of allocations) and the third in the general case.
The analysis and the results resembles to a great extent the theory 
of non-cooperative dynamic games.

\bigskip \bigskip

\noindent AMS 2000 Subject Classification: Primary 91A12, Secondary 91A25.

\bigskip \bigskip

\noindent {\emph{Journal of Economic Literature} classification
numbers:} C71.
\bigskip\bigskip

\noindent Keywords: fair core, stable core, credible core, convexification of a game.

\end{abstract}

\thispagestyle{empty} \vfill

\pagebreak  %\setcounter{page}{1}

%\doublespacing

\section{Introduction}\label{sec intro}

Noncooperative game theory has dedicated a lot of attention to dynamic games
and refinements of  Nash equilibrium have been studied to
capture the features that the dynamic induces in the game.
When the dynamic is obtained by simply
repeating a stage game over time, the folk theorem shows that the set
of equilibria in an infinitely repeated game is in general much larger
than the set of equilibria in the stage game.

In cooperative game theory most of the literature studies only
static situations: a game is played only once and its solution is a
set of suitable allocations that satisfies some conditions.

In this paper we consider a \emph{bona fide} dynamic version of a
cooperative game, where the worth of coalitions varies over time
according to the history of the game. In particular the worth of
coalitions at time $t$ depends on the allocations at all the times
before $t$.

When defining a solution concept we allow the possibility for
coalitions to deviate at any time and thereby to give rise to a new
environment. When a coalition deviates, from that point on, the game
is no longer played with the original set of players. The deviating
coalition becomes the new grand coalition which, in turn, induces a
new dynamic game. The stage games of the new dynamical game depend
on all previous allocations, including those that have materialized
from the deviating time on.

The existing literature on dynamic cooperative games considers games
that determine only the worth of any coalition in a stage game
played with the original grand coalition. However, in order to
accommodate the possibility of deviating coalitions that generate
new dynamical games, we need a richer structure. In the model of
dynamic cooperative games that we introduce, the grand coalition of
any stage game might be strictly smaller than the original grand
coalition, while the allocation history is adapted accordingly.

In this paper we focus on dynamic games where the stage games are
deterministically determined by the historical allocations. In these
games a sequence of allocations uniquely induces a sequence of
stage games. We investigate the core in three different approaches.

A coalition is said to be under-treated if the present value of its
stage-shares is smaller than the present stage-worth of it. A
sequence of stage-allocations is in the \emph{fair core} if no
sub-coalition is under-treated. Under the fair core approach,
which is similar to that taken by the relevant literature, an
under-treated coalition may complain but it cannot change the
evolution of the game by abandoning the previous environment and
creating a whole new game.

In the \emph{stable core} on the other hand, the share of a
coalition is compared to the opportunities it would have if it decided to
deviate. A coalition is said to be dissatisfied with a sequence of
allocations if, by quitting the original game, the coalition can
form another dynamic game, with a smaller number of players, and
afford better future allocations. A sequence of allocations is in
the stable core if no coalition can deviate and get on its own a
greater share than the one proposed by the sequence.

The stable core does not consider what the threat of a potential
deviating coalition consists of. It might be that the sequence of
allocations a coalition shows in order to substantiate its
dissatisfaction, is itself prone to deviations. Thus, a threat of a
coalition to deviate and obtain a certain sequence of allocations
may be non-sustainable and therefore non-credible. The
\emph{credible core} requires that any better sequence it might
generate on its own be credible. That is, any sequence of future
allocations must itself be immune to deviations of smaller
sub-coalitions that are also immune to deviations of smaller
coalitions.

In all our analysis at every stage $t$, either no player deviates
and therefore a game involving all the players of the previous stage
$t-1$ is played, or a coalition deviates and creates its own game,
which is a subgame of the previous one.  Players are never allowed
to establish a larger coalition once they have deviated and formed a
smaller one. So new games can be created by splitting, but not by
aggregation.  We make this assumption since, without it the possible
dynamics would be so general as not to produce any interesting
result. Moreover the assumption allows to describe a huge spectrum
of situations of relevance.

When a coalition $S$ deviates from the grand coalition $N$, we do
not take into account what happens to the coalition $N \setminus S$.
This is due to the fact that we are concerned with stability and
therefore with conditions that guarantee that no deviation will
actually materialize, no matter what the status of the abandoned
coalition is.

\subsection{Existing literature}\label{sec literature}

Dynamic cooperative games have been
studied in a few versions. Most of the studies, as we do,
concentrate on the core. \citet{Ovi:EJOR2000} studies the core of a
finitely repeated discounted cooperative game where the stage game
does not vary over time and no dynamic consideration is involved.

\citet{KraPerPet:IGTR2005} consider a finite horizon of
predetermined games. They study three different core concepts. The
\emph{classical core} assumes that coalitions planning to split off
do so right at the beginning. This concept does not depend on the
temporal structure of the game: the classical core coincides with
the core of an induced static game. The \emph{strong sequential
core}, on the other hand, allows for deviations of coalitions at any
stage of the game, but once a coalition deviates at some point, it
must keep doing so from that time on. In the above two concepts
deviations are not required to be credible, i.e., they could be
blocked by some sub-coalition in the future. The \emph{weak
sequential core} is robust against credible coalitional deviations.
The latter means that deviations are immune to
deviations of sub-coalitions. The sub-coalition deviations can be
themselves non-credible. \citet{HabHer:mimeo2010} provide a correction of the above definition of weak sequential core.

\citet{Pre:GEB2007} deals with infinite-horizon stationary
cooperative games, where at each moment the game is in one of a
finite number of states, that determines which instantaneous game is
played at that moment. The states evolves according to an exogenous
Markov chain and it does not depend on past allocations.  The author
considers the classical core and a version of the strong sequential
core, and provides conditions for nonemptyness of the strong
sequential core.
\citet{Hel:MPRA2008} focuses on the bargaining set of dynamic
cooperative games, where the sequence of stage TU-games is
exogenously specified.

Related results can be found in \citet{Gal:JET1978}, where a concept
of \emph{sequential core} is defined and is used to model lack of
trust in a two-period economy. In this model coalitions are allowed
to deviate in the second period. \citet{BecCha:E1995} consider
infinite horizon capital allocation models and define
\emph{recursive core} allocations, the ones where no coalition can
improve upon its consumption stream at any time given its
accumulation of assets up to that period.

\citet{Kou:ET1998}  introduces the notion of \emph{two-stage core},
that takes into account the possibility of temporary cooperation.
Within each coalition agents make future trades only if they are
enforceable, i.e., a coalition may have a limited horizon. Moreover
a coalition blocks at some point in time only if it can secure
improvements for its members in any possible consequence of a
deviation. \citet{PreHerPet:JME2002, PreHerPer:IJGT2006} use the
concepts of strong and weak sequential core in the context of
two-period economies. \citet{PreHerPet:ET2004} apply the concept of
strong sequential core to a stationary exchange economy.

\citet{Pet:VLU1977, Pet:World1993} deals with a cooperative game
induced by a (non-cooperative) differential game and
\citet{PetZac:JEDC2003} study the problem of allocation over time of
total cost incurred by countries in a cooperative game of pollution
reduction and compute the Shapley value of this game. These papers
are not about dynamic cooperative games but about a cooperative game
induced by non-cooperative game played over time.

There exists a whole literature on coalition formation, where
stability of coalitions is considered under different aspects
\citep[see, e.g.,][and references therein]{Ray:IJGT1989,
Chw:JET1994, Xue:ET1998, RayVoh:GEB1999, KonRay:JET2003,
DiaXue:JET2007, Ray:Oxford2007}. Typically this literature considers
strategically richer models than the one examined in this paper, so
it is closer in spirit to noncooperative game theory. For instance
the model considered by \citet{KonRay:JET2003}, which describes
coalition formation as a truly dynamical process, considers a state
space, beliefs, a probabilistic structure, and equilibrium concepts.

Our notion of credible core can be related to the papers by
\citet{BerPelWhi:JET1987, BerWhi:JET1987} on coalition-proof Nash
equilibria.

In some of our results we resort to the concept of
$\varepsilon$-core. This was introduced in \citet{ShaShu:E1966} to
analyze situations where the core is empty. It has been employed in
different  contexts by \citet{Woo:JME1983,  ShuWoo:MSS1983a,
ShuWoo:MSS1983b, WooZam:E1984, KovWoo:GEB2001, KovWoo:MOR2001,
KovWoo:JET2003, KovWoo:ET2005}, among others. In some of these
papers a parametrized collection of cooperative games is considered
and an approximate core is computed, where the goodness of the
approximation depends on the parameters of the game. In particular,
given the parameters $\pi$ describing a collection of games and
given a lower bound $n_{0}$ on the number of players in each game in
the collection, \citet{KovWoo:MOR2001} obtain a bound
$\varepsilon(\pi, n_{0})$ so that, for any $\varepsilon \ge
\varepsilon(\pi, n_{0})$, all games in the collection with at least
$n_{0}$ players have nonempty $\varepsilon$-cores. Some of our
results have a similar flavor, except that for us the lower bound on
the $\varepsilon$ is zero and the quantity  that guarantees the
existence of the $\varepsilon$-core is the discount factor, rather
than the number of players.

The paper is organized as follows. In Section \ref{sec motivating}
we give a motivating example based on the classical market games of
Shapley and Shubik.  The model is
introduced in \ref{sec dynamic-games} and the two first types of
core solutions are given in Section \ref{sec core}. Section \ref{sec
Markov-core} provides characterizations of the non-emptyness of the
$\varepsilon$-core when the discount factor is sufficiently large. The
credible core is discussed in Section \ref{sec def-core-credible}
and the paper ends with a section devoted to a few final
remarks.

\section{A motivating example: A market with externalities} \label{sec motivating}
To show a typical application of dynamic cooperative games consider
$n$ firms that engage repeatedly in a market game. At any period
each firm brings into the market its own endowment and technology,
that might depend on the firm's previous allocation. The firms then
share their endowments in order to produce the maximal  possible
quantity. An important feature of the model is the existence of a
positive externality reflected in the knowhow of each firm. The
production function of each firm increases as the number of firms in
the economy increases.

Formally, let $N=\{1,2, \dots, n\}$ be the set of firms that are
capable of producing a certain commodity using $\ell$ production
factors. In the static version of the model, when the input of
production factors is $y=(y_1,\dots, y_\ell)\in
\mathbb{R}^{\ell}_+$, firm $i$ produces $e(n)u^i(y)$, where $u^{i}$
is a concave function and $e(n)$ is the externality factor which is
increasing with the number  $n$ of firms in the economy.

The relevance of the externality factor becomes clear in the dynamic
model. An under-treated coalition of firms might want to split off
and form its own consortium. By doing so, on one hand, as an
independent consortium, it will be subject to a smaller externality
factor, since the number of cooperating firms is reduced. On the
other hand, it will have the full freedom to share the entire profit
the way it wishes.

To make the model more realistic, we assume that the production
functions change over time and that, in order to keep the production
ability, firms need to invest every period in maintenance, which
requires resources. These resources come from the allocation of the
firms in previous times, and whatever does not go into maintenance,
is used for dividends. Thus, the current production function depends
on yesterday's allocation and the externality factor.

For instance, suppose that the production function of firm $i$ at
time $t$ is
\begin{equation}\label{eq prod-func}
 u^i_t(y)= e(k)\gamma^{1/(1+x^i_{t-1})} u^i_{t-1}(y),
\end{equation}
 where
$k$ is the number of firms in the consortium that $i$ belongs to,
$x^i_{t-1}$ is the allocation of firm $i$ at time $t-1$, and
$0<\gamma<1$ is the decay rate per period. Note that
$\gamma^{1/(1+x^i_{t-1})}$ is increasing with $x^i_{t-1}$ and
therefore the greater the allocation at time $t-1$, the more
efficient the firm at time $t$.

For the sake of simplicity assume that each firm is endowed anew at
any time with the same production factor basket, say $y^i$. We are
ready now to describe the dynamic. If at time $t$  firm $i$ belongs
to the consortium $S$, then it engages the stage market game
$v^S_t$ defined by
\[
v^S_t(T)=\max\Big\{\sum_{i\in T}u^i_t(z^i); ~ \sum_{i\in T}
z^i=\sum_{i\in T}y^i, ~z^i\in \mathbb{R}^{\ell}_+\Big\}
\]
for every $T\subseteq S$, with $k=|S|$ in \eqref{eq prod-func}.

\section{Dynamic cooperative games}
\label{sec dynamic-games}
\subsection{Dynamic of the game}\label{subsec dynamic}
Let $N=\{1,\dots ,n\}$ be the set of players. For any coalition
$S\subseteq N$ consider a
function $v^S: 2^S \to \mathbb{R}_+$ with $v^S(\varnothing)=0$. The
function $v^S$ is called  \emph{characteristic function defined over
$S$} with the set $S$ being the \emph{grand coalition} of $v^S$. An
\emph{allocation} of $v^S$ is a vector $x^S\in \mathbb{R}^S$ that
satisfies $\sum_{i\in S}x^S(i)=v^S(S)$ and $x^S(i)\ge B$, where $B$
is a uniform lower bound over all allocations. The reason for this
lower bound is primarily technical: with this lower bound the set of
allocations becomes compact. If we take $B=0$,  no inter-temporal
loans are allowed, whereas, when $B<0$ a player can get less than
her individually rational level at a certain stage, but then she
will be compensated in the future.

At any stage $t$ a cooperative games over a grand coalition $S$ is
played. Both the game and the grand coalition depend on the history
up to that stage. The players of the grand coalition $S$ are getting
at time $t$ an allocation of the game  actually being played. The
cooperative game of the subsequent period depends on the current
allocation.

Formally, a dynamic cooperative game is played over a discrete set
of periods. The evolution of the system depends on the initial game,
the allocation at every period and the dynamic $V$, specified below.
At stage $1$ any coalition decides whether to split off or not. If
no coalition splits off, then the initial cooperative game $v_1^N$
is played, and the allocation is $x^N_1$. If $S$ splits off,  then
the initial cooperative game is $v_1^S$ and the allocation is
$x^S_1$. Note that the initial game is predetermined and is beyond
the control of the players, unless a sub-coalition wishes to split
off.

Like in a static cooperative game we will use (different versions
of) the core to determine whether the grand coalition $N$ is stable
(in different senses), namely no sub-coalition $S$ has an incentive
to split. In case a coalition $S$ can advantageously split, the
grand coalition $N$ is not stable, which makes the core empty. That
is why we do not need to specify what payoffs players in $N\setminus
S$ get or whether some of them want to form a coalition of their
own.

From period 2 on, the evolution is governed by $V$. The state of
the system is a pair $(S; x^S)$, where $S$ is a coalition and $x^S$
is an allocation of $S$. As long as no coalition splits off, the
state is of the form $(N; x^N)$; once a coalition $S\subseteqq
N$ deviates, the system turns to a state of the type $(S; x^S)$, and
$S$ remains fixed from that stage on forever.

When the state at time $t$ is $(N; x^N)$, unless a deviation of $S$
occurs, the game played at time $t+1$ is $V(N; x^N)$, whose grand
coalition is $N$. At this time the allocation is $x^N_{t+1}$.
However, if $S$ deviates, the game played at time $t+1$ is $V(S;
x^N_S)$, where $x^N_S$ is the allocation induced by  $x^N$ to
coalition $S$. At time $t+1$ the allocation is $x^S_{t+1}$ and the
subsequent game is  $V(S;x^S_{t+1})$ whose grand coalition is $S$.

For the sake of simplicity we assume a Markovian structure of the
game, where the stage game played at time $t$ depends on the
allocation at time $t-1$. More complicated dynamics could be
considered, for instance the game at time $t$ could depend on the
whole past history. An interesting intermediate case is the one
where the game depends on some unidimensional function of the
history, for instance on the sum of the past allocations. Think for
instance of a model of dynamic public good provision, where every
player contributes to a public good, whose level at time $t$ depends
on the (discounted) sum of past contributions.

A few remarks regarding our modeling choices are in place.  Up to 
Section \ref{sec def-core-credible}, we assume that once a coalition
deviates, it remains the grand coalition forever and no further
splitting off of sub-coalitions will take place. This restriction
corresponds to the first two types of core solutions, that are
concerned with long-term plans that prevent these kind of
deviations. When dealing with the third type of core solution 
we lift this restriction. A deviating coalition is not protected
against coups of its sub-coalitions. This is the reason why a threat
of a coalition to deviate is rendered credible only if it is
immunized against further split offs of its sub-coalitions.

We also assume here that once a coalition deviates, there will be no
way to restore a full cooperation and to rebuild the grand coalition
$N$. Such a possibility requires a much more complicated dynamic
that would depend also on past allocations of non-deviating members.
In this paper we decide to keep matters as simple as possible. This
distinguishes our model from the literature on coalition formation
that we mentioned in the Introduction.

In our model a coalition $S$ can deviate prior to time $1$ and play
the game $v_1^S$. We could as well assume that at time $1$ for any
$S$, $V(N, x^{N}_{0})(S)=V(S, x^{N}_{0})(S)$, and it will make no
difference in the results. Note however, that in the motivating
example, when the grand coalition is $N$, the worth of $S$ is
typically different from its worth when the grand coalition is $S$.

Our model refers to the deviating coalition, but it ignores the rest 
of the players. In principle, the complement of a deviating
coalition could be treated just like the deviating coalition itself.
The continuation game of the complement depends on its historical
allocations. However, we chose not to specify it for two reasons.
First, our study focuses on two main aspects, fairness and
stability. Whether or not a coalition is treated equitably does not depend 
on what happens to its complement. The same applies to the
willingness of a coalition to split off: it is not affected by the
its complement.

The second reason is that we study core solutions and characterize 
the games whose core is non-empty. In these games no deviation will
occur and in any case no coalition will be left alone without its
complement.

\subsection{Discounting future payoffs}\label{subsec discount}

We assume throughout that all players have the same discount factor
$0<\delta<1$. Suppose that player $i$'s payoff at time $t$ is
$x_t(i)$. Her present normalized payoff is
\[
x_*(i,\delta)=(1-\delta)\sum_{t=1}^\infty \delta^{t-1} x_t(i).
\]
We define for every $T\subseteq N$
\[
{x}_*(T,\delta)=\sum_{i\in T}{x}_*(i,\delta).
\]
Therefore ${x}_*(T,\delta)$ is the sum of all individual allocations
of  $T$'s members.

Define a new characteristic function on $N$ as follows.
\[
v_{*}(S,\delta)= \max {x}_*(S,\delta),
\]
where the maximum is taken over
all feasible histories of $S$-allocations: $x^S_1,x^S_2,\dots $ and
$x_t(i)=x^S_t(i)$.

The dynamic game with discount factor $\delta$ will be denoted by
$(V,\delta)$.

\section{The fair and stable core solutions}\label{sec core}

\subsection{The fair core}\label{subsec def-core-fairness}

There are two justifications for the definition of core in the
classical model of one-shot cooperative game. The first is fairness
and the second is stability. An allocation is in the core if any
coalition obtains at least its worth. Behind this justification lies
an assumptions that a central planner has a full control on what the
players get, and once she makes up her mind regarding the split of
the cake, the players have no way to protest.

This kind of reasoning leads us to define what we call the `fair
core' first. A sequence of allocations $x^N_1,x^N_2,\dots  $ is in
the fair core if fore every coalition $S$ the present value of the shares of $S$
exceeds the present value of the  worths of
$S$. It is assumed that coalition $S$ can  do nothing about its
future shares, which gives the central planner the freedom to choose
allocations without paying attention to semi-strategic considerations
like stability.

The definition of fair core is concerned solely with the
following consideration: it is fair to give coalition $S$
allocations whose present value is no less than what their present
worth is.

\begin{definition}[Fair core]
\begin{enumerate}[{\rm (i)}]
\item A sequence of $N$-allocations  $x^N_1,x^N_2,\dots  $ is in the
\emph{fair core} of $(V,\delta)$ if for every $S\subseteq N$
\[
{x}^N_*(S,\delta)\geq (1-\delta)\sum_{t=1}^\infty \delta^{t-1}
V(N;x^N_{t-1})(S).
\]

\item A sequence
of $N$-allocations $x^N_1,x^N_2,\dots  $ is in the\emph{
$\varepsilon$-fair core} of $(V,\delta)$ if for every $S\subseteq N$
\[
{x}^N_*(S,\delta)\geq (1-\delta)\sum_{t=1}^\infty \delta^{t-1}
V(N;x^N_{t-1})(S)-\varepsilon.
\]
\end{enumerate}
\end{definition}

If the core of the stage game at time $t$ is nonempty, for
every $t$, then the fair core of $(V,\delta)$ is nonempty.

\begin{example}
It is possible for a game $(V,\delta)$ to have a nonempty fair core
even if for every $t$ the core at the stage game is empty. Consider
the four player games $u_1$ and $u_2$, where $u_1(i)=0$ for every
player $i$, $u_1(12)=u_1(23)=u_1(13)=3$, $u_1(123)=4$ and player 4
is dummy; game $u_2$ is like the $u_1$ where players 1 and 4
exchange their roles. The cores of $u_1$ and $u_2$ are empty.

Suppose that $v_{1}^{N}=u_1 $,
$V(N;(1,\frac{3}{2},\frac{3}{2},0))=u_2$ and
$V(N;(0,\frac{3}{2},\frac{3}{2},1))=u_1$. Suppose that
$x^N_t=(1,\frac{3}{2},\frac{3}{2},0)$ when $t$ is odd and
$x^N_t=(0,\frac{3}{2},\frac{3}{2},1)$ when $t$ is even. The result
is that $u_1$ is played in odd times and $u_2$ in even times.

When the discount factor $\delta$ is high, the average over time of
stage games is close to $(u_1+u_2) /2$ and the discounted value of
the stream of allocations is close to
$(\frac{1}{2},\frac{3}{2},\frac{3}{2},\frac{1}{2})$. Thus the
sequence of payoffs $(x^N_t)_t$ is in the $\varepsilon$-fair core of
$(V,\delta)$.
\end{example}

\begin{definition}[Efficiency]
A sequence of $N$-allocations  $x^N_1,x^N_2,\dots  $ is
\emph{efficient} in $(V,\delta)$ if
\[
{x}^N_*(N,\delta)\geq v_{*}(N,\delta).
\]
\end{definition}
That is,  $x^N_1,x^N_2,\dots  $ is efficient if the present value of
the grand coalition's share is the maximum available. Recall that
the actual game is history dependent. While any of the stage
allocations could be locally efficient, it might reduce the size of
the cake in subsequent periods and thereby might hamper efficiency.

In the classical one-shot game the definition of allocation
contains the requirement of efficiency. This is not the case in the
dynamic game. The fair core is not necessarily efficient, as
demonstrated by the following example.

\begin{example}
Suppose that
\[
v_{1}(N;S)=
\begin{cases}
1 & S=N, \\
0 & \text{otherwise.}
\end{cases}
\]
and, if the allocation $x^N$ is uniform (i.e., treats all players
equally), then $ V(N;x^N)= |x^N|v_{1}^{N}$, where $|x^N|$ stands for
the sum of the individual allocations of all players, otherwise $
V(N;x^N)= 0$. Call $e_{1} = (1,0,\dots, 0) \in \mathbb{R}^{N}$ the
first vector of the standard basis and $\overrightarrow{0} =
(0,\dots, 0) \in \mathbb{R}^{N}$ the zero vector.   The sequence of
allocations $e_{1}, \overrightarrow{0}, \overrightarrow{0},\dots $
is in the fair core: the first allocation $e_{1}$ is in the core of
the stage game $v_{1}^{N}$, but it is not uniform and therefore all
subsequent games are identically $0$. This sequence is in the fair
core of  $ (V,\delta)$, but it is certainly not efficient.
\end{example}

\begin{definition}[Efficient fair core]
\begin{enumerate}[{\rm (i)}]
\item A sequence of $N$-allocations  $x^N_1,x^N_2,\dots  $ is in the
\emph{efficient fair core} of $(V,\delta)$ if it is efficient
and in the fair core of $(V,\delta)$.

\item A sequence
of $N$-allocations $x^N_1,x^N_2,\dots  $ is in the\emph{
$\varepsilon$-efficient fair core} of $(V,\delta)$ if it is efficient
and in the $\varepsilon$-fair core of $(V,\delta)$.
\end{enumerate}
\end{definition}

In the special case where the worth of the grand coalition is
constant, that is when  the worth of the grand coalition in all
stage games do not depend on the history of allocation nor on the
time, the fair core and the efficient fair core coincide.

\subsection{The stable core}\label{subsec def-core-stability}
The definition of the fair core does not make use of the entire
structure of the dynamic game. It uses only states of the type $(N;
x^N)$ and not of the type $(S; x^S)$, where $S\subsetneqq N $.

The second justification of the core involves stability
considerations. An under-treated coalition might deviate, create its
own game, and improve its position by reallocating its endowment.
When a coalition $S$ threatens to deviate, it shifts the system to a
state of the form $(S; x^S)$ and the dynamic then is governed by
$V(S; \cdot)$. This is why referring just to states of the type $(N;
x^N)$ is insufficient and there is a need to refer to states of the
form $(S; x^S)$ for every coalition $S$.

Let $x^N_1, x^N_2, \dots$ be a sequence of allocations. Define
\[
x^h_*(i,\delta)=(1-\delta)\sum_{t=h}^\infty \delta^{t-1} x_t(i).
\]
The number $x^h_*(i,\delta)$ represents the discounted value at time $h$ of
the shares that player $i$ receives at time $h$ and on. In
particular, $x_*(i,\delta)=x^1_*(i,\delta)$. Similarly, define the
game $v^h_{*}(\cdot,\delta)$ as
\[
v^h_{*}(S,\delta)=\max \sum_{i \in S}
x^h_*(i,\delta),
\]
where the maximum is taken over all feasible sequences $x^N_1,
x^N_2, \dots x^N_{h}, z^S_{h+1},z^S_{h+2}\dots$ that coincide with
$x^N_1, x^N_2, \dots$ up to $h$, the time when $S$ decides to leave.

A sequence of allocations $x^N_1,x^N_2,\dots $ is in the stable
core if  at any time $h$ the value of the shares of coalition $S$
exceeds the value of what coalition $S$ could guarantee in autarky,
meaning without being engaged with others.

\begin{definition}[Stable core]
\begin{enumerate}[{\rm (i)}]
\item
A sequence of $N$-allocations  $x_1,x_2,\dots  $ is in the
\emph{stable core} of $(V,\delta)$ if for every $S\subseteq N$
and every $h$,
\[
{x}^h_*(S,\delta)\geq v^h_{*}(S,\delta).
\]
\item
 A sequence of $N$-allocations
$x^N_1,x^N_2,\dots  $ is in the \emph{$\varepsilon$-stable core} of
$(V,\delta)$ if for every $S\subseteq N$ and every $h$,
\[
{x}^h_*(S,\delta)\geq v^h_{*}(S,\delta)-\varepsilon.
\]
\end{enumerate}
\end{definition}

\begin{remark}
\begin{enumerate}[{\rm (a)}]
\item Unlike the fair core, any
history of $N$-allocations in the stable core is efficient.

\item The two notions of core are not necessarily co-variant
with linear transformations. If $V$ is co-variant with linear
transformations, so are the two cores.
\end{enumerate}
\end{remark}

\begin{example}
It is possible for a game $(V,\delta)$ to have a nonempty stable
core even if for every $t$ the core of the stage game is empty.
Consider a set $N = \{1, 2, \dots, 2k+1\}$, and let
\[
V(N;x^N)(S) =
\begin{cases}
x^N(N) & \text{ if } |S| \ge k+1,\\
0 & \text{ otherwise.}
\end{cases}
\]
be $x^N(N)$ times a majority game, and let $ V(S;x^S)(T) =2^{-|T|}
V(N;\cdot)(T) $ for every $S\subsetneqq N$.

The core of each stage game $V(N; x^N_t)$ is empty, and so is the
fair core of $(V,\delta)$, whereas its stable core is not,
regardless of the discount factor. The reason being that when a
coalition deviates, its future payoff declines rapidly. Coalitions
will be satisfied with shares that are strictly smaller than their
worth, because deviation does not promise a greater portion.
\end{example}

\section{Non-emptyness of the core}\label{sec Markov-core}

\subsection{Non-emptyness of the fair core}\label{subsec
fairness-Markov-core} The following result applies to  games where
$v(N;x^N)(N)$ is equal to $x^N(N)$ for every allocation $x^N$.

For the next definition we recall that the set of characteristic
functions is a vector space. Denote by $\Delta(d)$ the set of stage
allocations of games where the worth of the grand coalition is $d$.
That is, $\Delta(d)=\{x^N; x^N(N)=d, x^N(i)\ge B$ for every $i\in N
\}$. The set $\Delta(d)$ is obviously compact, which is the reason why we
impose the constraint that  $x^N(i)\ge B$ for every $i\in N$.

Let $x\in \Delta(d)$ be an allocation. We say that
$(y_1,\dots,y_k;\alpha_1,\dots,\alpha_k)$ is a \emph{split} of $x$,
if
\begin{equation}\label{eq1}
  x=\sum_{j=1}^k \alpha_j y_j,
\end{equation}
where $\alpha_j\ge 0$,  $\sum_{j=1}^k \alpha_j = 1$, and $y_j$ is an
allocation $j=1,\dots, k$. That is, $x$ is a convex combination of
the allocations $y_j\in \Delta(d)$ with $\alpha_j$ being the
respective weights.

\begin{definition}  The \emph{convexification}
of $V(N; \cdot)$, denoted $\conv V(N; \cdot)$,
is a correspondence defined as follows. Let $x\in \Delta(d)$ be an
allocation. Then, $\conv V(N;x)$ is the set of all games that can be
expressed as $\sum_{j=1}^k \alpha_j V(N,y_j)$, where
$(y_1,\dots,y_k;\alpha_1,\dots,\alpha_k)$ is a {split} of $x$.
\end{definition}

If $V(N;\cdot )$ is continuous, then  $\conv V(N;x)$ is closed and
therefore, $\conv V(N;x)$ also contains all games of the form
$\sum_{j=1}^\infty \alpha_j V(N;y_j)$, where $x$ is expressed as an
infinite convex combination of allocations: $
  x=\sum_{j=1}^\infty \alpha_j y_j.
$

\begin{theorem}\label{th:Markov}
Consider a game where $V(N,\cdot)(N)=d$ and $V(N;\cdot)$ is
continuous. Assume that for every coalition $S$, $V(N;\cdot)(S)$ is
bounded. For every $\varepsilon>0 $ there is $0 < \delta_0 < 1$ such that
for every $\delta\in (\delta_0,1)$ the $\varepsilon$-fair core of
$(V,\delta)$ is not empty if and only if there exists $x\in
\Delta(d)$ and $v\in \conv V(N;x)$ such that $x$ is in the core of
$v.$
\end{theorem}
Before we get to the proof we wish to comment on the contents of this %Ehud 15/3
theorem. Just like in the folk theorem of the non-cooperative game
theory, it characterizes the solution of the dynamic game in static
terms. Specifically, it characterizes when the $\varepsilon$-fair
core of the dynamic game is not empty in terms of the
{convexification} of $V(N; \cdot)$.

Note that we refer to the $\varepsilon$-fair core rather than to the 
fair core. The question is whether we do it because we just cannot
prove anything stronger, or that it is due to a structural
insurmountable difficulty. Recall that the dynamic game is described
by the dynamics, $V$ and by an initial game played at the first
stage. While there is some control of future games through past
allocations, there is no way to alter the initial game. It might
happen that this stage-game hinders the existence of an exact core
while  $\varepsilon$-fair core does exist. 

\begin{proof}[Proof of Theorem~\ref{th:Markov}]
Suppose first that for every $\varepsilon>0 $ there is $1>\delta_0>0 $ such
that for every $\delta \in (\delta_0,1)$ there is a sequence
$x_1,x_2, .. .$ of allocations in the $\varepsilon$-fair core of
$(V,\delta)$.

Denoting,
\[x=(1-\delta)\sum_{t=1}^\infty \delta^{t-1} x_t \quad\text{and}\quad
u=(1-\delta)\sum_{t=1}^\infty \delta^{t-1} V(N;x_{t-1}),
\]
we have for every $S\subseteq N$,
\[
x(S)\geq (1-\delta)\sum_{t=1}^\infty \delta^{t-1}
V(N;x_{t-1})(S)-\varepsilon=u(S)-\varepsilon.
\]
Here, $V(N;x_{0})$ denotes the initial game, $v_1$.

Note that $(1-\delta)\sum_{t=2}^\infty \delta^{t-1} V(N;x_{t-1})=
u(S)-(1-\delta) v_1$, and we obtain
\begin{align*} x(S) &\geq
u(S)-(1-\delta)v_1(S)+(1-\delta)v_1(S)-\varepsilon \\
&= (1-\delta)\sum_{t=2}^\infty \delta^{t-1} V(N;x_{t-1}) +
(1-\delta)v_1(S)- \varepsilon\\ & \ge (1-\delta)\sum_{t=2}^\infty
\delta^{t-1} V(N;x_{t-1}) - \varepsilon.
\end{align*}

When $\delta$ is sufficiently close to $1$, since $V(N;\cdot)(S)$ is
bounded, we have
\[
\delta\frac{1-\delta}{\delta}\sum_{t=2}^\infty
\delta^{t-1} V(N;x_{t-1})(S)>
\frac{1-\delta}{\delta}\sum_{t=2}^\infty \delta^{t-1}
V(N;x_{t-1})(S)- \varepsilon.
\]
Thus,
\begin{align*} x(S) &\geq (1-\delta)\sum_{t=2}^\infty \delta^{t-1}
V(N;x_{t-1})-\varepsilon \\
&=\delta\frac{1-\delta}{\delta}\sum_{t=2}^\infty
\delta^{t-1} V(N;x_{t-1})(S)-\varepsilon \\
&\geq  \frac{1-\delta}{\delta}\sum_{t=2}^\infty \delta^{t-1}
V(N;x_{t-1})(S)-2\varepsilon.
\end{align*}

Define
\[
x'=\frac{1-\delta}{\delta}\sum_{t=2}^\infty \delta^{t-1}
x_t.
\]
If $\delta$ is sufficiently large, for every $S$,
$|x(S)-x'(S)|<\varepsilon$,
and therefore,
\[
x'(S)\geq\frac{1-\delta}{\delta}\sum_{t=2}^\infty \delta^{t-1}
V(N;x_{t-1})(S)-3\varepsilon.
\]

In other words, $x'$ is in the $3\varepsilon$-core of the game
\[
\frac{1-\delta}{\delta}\sum_{t=2}^\infty \delta^{t-1} V(N;x_{t-1})
\]
which is a (infinite) convex combination of the games $V(N;x_{t-1}),
t=2,\dots $, each with the weight
\[
\delta^{t-1}
\frac{1-\delta}{\delta}=\delta^{t-2}(1-\delta)
\]
and is therefore in $\conv V(N;x')$ (because $x'$ is a convex
combination of $x_{t-1}$'s with the weights
$\delta^{t-2}(1-\delta)$, $t=2,3,\dots $).

Since $\varepsilon$ is arbitrary, from compactness and continuity of
$V(N;\cdot)$, we conclude that there exists $v\in \conv V(N;x)$ such
that $x$ is in the core $v$, as desired.

We now assume that there is a vector $x\in \Delta$ such that $x$ is
in the core of $v\in \conv V(N;x).$ By definition of $\conv V(N;x)$,
there is a split $(y_1,\dots,y_k;\alpha_1,\dots,\alpha_k)$ of $x$
such that
\begin{equation}\label{eq2}
\alpha_1V(N;y_1)+\dots +\alpha_kV(N;y_k)=v.
\end{equation}
Fix an $\varepsilon>0$. For $\delta$ large enough one can divide the set of
periods into $k$ disjoint sets $T^1,\dots ,T^k$ in a way that for
every $j=1,\dots ,k$,
\begin{equation}\label{eq:alphadelta}
\alpha^j=(1-\delta)\sum_{t\in T^j}\delta^{t-1}
\end{equation}
(see for instance, \citet{FudMas:E1986}).

At time 1 set $v_1$ as an arbitrary game where $v_1(N)=d$.
And in general, for $t\in T^j$ define $x_t=y_j$. In words, over the
time periods in the set $T^j$ the allocation is $y_j$ and the game
that follows is $V(N;y_j)$. Note that since $V(N;y_j)(N)$ is fixed
and equal to $d$, for every $i=1,\dots, k$,  $y_i$ is an allocation
of $V(N;y_j)$ (because $y_i \in \Delta(d)$). By construction and \eqref{eq1} the
present value of allocations is $x$.

On the other hand, the present value of all the stage games is
\[
(1-\delta)v_1+\sum_{j=1}^k(1-\delta)\sum_{t\in T^j, t\not =
1}\delta^{t-1}V(N;x_{t-1})=
(1-\delta)v_1+\sum_{j=1}^k(1-\delta)\left[\sum_{t\in T^j, t\not =
1}\delta^{t-1}\right]V(N;y_j).
\]

Let $j_0$ be such that $1\in T^{j_0}$. Using \eqref{eq:alphadelta}
we obtain for every coalition $S$,
\begin{align*}
&(1-\delta)v_1(S)+\sum_{j=1}^k(1-\delta)\sum_{t\in T^j, t\not =
1}\delta^{t-1}V(N;x_{t-1})(S) \\
&\qquad\qquad=(1-\delta)v_1(S)+ \sum_{j=1}^k\alpha_j V(N;y_j)(S)-
(1-\delta)V(N;y_{j_0})(S)\\
&\qquad\qquad\le(1-\delta)v_1(S)+ \sum_{j=1}^k\alpha_j V(N;y_j)(S).
\end{align*}
As $V(N; \cdot)$ is continuous,
$V(N;x)(S)$ is bounded and therefore, when $\delta$ is sufficiently
close to 1, $(1-\delta)v_1(S)<\varepsilon
$ for every coalition $S$. Thus, the present value for every
coalition $S$ satisfies,
\[
\Big((1-\delta)v_1+\sum_{j=1}^k(1-\delta)\sum_{t\in T^j, t\not =
1}\delta^{t-1}V(N;x_{t-1})\Big)(S)\le  v(S) +\varepsilon.
\]
Since $x$ is in the core of $v$, the sequence $x_1,x_2, \dots$ is in
the $\varepsilon$-fair core of the dynamic game $V^N$ with the
discount factor $\delta$. \qed
\end{proof}

\begin{example}

Let $N = \{1,2,3\}$ and let $e_i$ be the $i$-th vector of the
standard basis in $\mathbb{R}^{3}$, that is, its $i$-th coordinate
is $1$ and the others are $0$. Define  $ V(N,e_i)$ to be the
additive game where $v(j)=1/2$ for $j\not = i$ and $v(i)=0$, which
we denote as $p_{-i}$. When $x=(x^{1}, x^{2}, x^{3})$ is such that
$x^i\le 4/5$ for every $i\in N$, then  $ V(N,x)$ is the simple
majority game (i.e., $V(N, x^{N}_{0})(S) = 1$ iff $|S| \ge 2$).
Moreover $ V(N,\cdot)$ is extended to the whole simplex in a
continuous fashion, keeping $ V(N,\cdot)(N)=1$. Note that in all the
games involved, the feasible allocations are elements of the simplex
$ \{(x^{1}, x^{2}, x^{3}): x^{i} \ge 0, \sum_{i} x^{i}=1\}. $

Let $v_1$, the initial game,  be the simple majority game. Set,
$x_1=e_1$, $x_2=e_2$ and\footnote{Here, $3k(\hskip-2mm \mod 3)=3$.}
$x_t=e_{t(\hskip-2mm \mod 3)}$. The dynamic induces: $v_2=V(N,
x_1)=p_{-1}$,  $v_3=V(N, x_2)=p_{-2}$, and  $v_t=V(N,
x_2)=p_{-(t-1)(\hskip-2mm \mod 3)}$. It turns out that when $t>1$
the core of  $v_t$ is non-empty. However, if the allocation at time
$t$ is the unique core allocation of $v_t$, the next stage-game is
the majority game, whose core is empty.

It is easy to check that $x_1,x_2, \dots  $ is in the
$\varepsilon$-fair core for discount factors large enough. To
see that the dynamical game satisfies the sufficient condition of
Theorem \ref{th:Markov}, consider $x=(1/3,1/3, 1/3)$. Let $v$ be the
additive game with weights $1/3$ assigned to each player. Note that
\begin{enumerate}[(a)]
\item
$x$ is in the core of $v$;
\item
$\sum_{i=1}^3 \frac{1}{3}e_i$  is a split of $x$ and
\item
$v=    \sum_{i=1}^3 \frac{1}{3}p_{-i} = \sum_{i=1}^3 \frac{1}{3}V(N,e_i)$.
\end{enumerate}
Thus the sufficient condition of Theorem \ref{th:Markov} is satisfied.
\end{example}

Let $M$ be a closed set of allocations. For a vector $x$ define
$\conv_M  V(N;x)$ as $\conv V^{N}(x)$ was defined, with the extra
condition that the allocations $y_j$ are in $\conv M$. That is,
$\conv_M V(N;x)$ is the set all games that can be expressed as
$\sum_{j=1}^k \alpha_j V(N;y_j)$, where
$(y_1,\dots,y_k;\alpha_1,\dots,\alpha_k)$ is a {split} of $x$ and
$y_j \in \conv M$, $j=1,\dots, k$.

\begin{theorem}\label{th efficinet fairness} Consider a game where $V(N;\cdot)$ is continuous and bounded.
For $\gamma>0$ denote $M_{\gamma} = \{x;~V(N;x)(N)>\sup_y
V(N;y)(N)-\gamma\}$. For any $\varepsilon>0$ the efficient $\varepsilon$-fair
core of a game $V(N;\cdot)$ is not empty for $\delta$ large enough
if and only if for every $\gamma$ sufficiently small there exists
$x$ and $v\in \conv_{M_{\gamma}} V(N;x)$ such that $x$ is in the
$\gamma$-core of $v$.
\end{theorem}

The proof is similar to the proof of Theorem~\ref{th:Markov} and is
therefore omitted.

\subsection{Non-emptyness of the stable core}\label{subsec
stability-Markov-core}

Recall that  $V(S;x^S)$ is the stage game played with $S$ as the
grand coalition after a stage in which the allocation was $x^S$. We
now assume that $V(S;x)(T)$ depends on $x^S(T)=\sum_{i\in T}x_i$ for
every $S$ that contains $T$ in a continuous and monotonically
increasing fashion. In particular, the worth of coalition
$T\subseteq S$ at time $t$ depends only on its total share at time
$t-1$. For every coalition $T$ and time $t$, we define $U^{t}_T(c)$,
inductively. $U^{ 1}_T(c)=V(S;x)(T)$, where $x(T)=c$. Note that this
is well defined, as $V(S;x)(T)$ depends solely on $x(T)$. Then, $U^{
t}_T(c)=U^{ 1}_T(U^{t-1}_T(c))$. Define $f_T(c)$ to be the limit of
$U^{t}_T(c)$. Due to continuity this limit exists. Thus, it
satisfies $f_T(c)= f_T(f_T(c))$. That is, $f_T(c)$ is a fixed point
of $U^1_T$ and of $f_T$.

We further assume that $f_T(c)$ is finite for every $T$ and $c$,
which in equivalent to assuming that either the set of fixed points
of $U^1_T$ is unbounded or $U^1_T(x)<x$ for $x$ sufficiently
large.\footnote{The function $U^1_T$ is continuous and monotonic. In
case the set of the fixed points of $U^1_T$ is unbounded, every non-fixed point of $U^1_T$ is between two fixed points. The set of fixed
points of $U^1_T$ is closed, and therefore for every non-fixed point
of $U^1_T$, say $c$, there are two closest fixed point, one above $c$ and one below it. The sequence $U^{ t}_T(c)$ then converges
to one of the two (depending on whether  $U^{ t}_T(c)>c$ or  $U^{
t}_T(c)<c$). If, however,  the set of fixed points of $U^1_T$ is
bounded, then the sequence $U^{ t}_T(c)$ diverges to infinity in
case $U^{ t}_T(c)>c$ asymptotically. In case  $U^{ t}_T(c)<c$
asymptotically, the sequence $U^{ t}_T(c)$ is decreasing and
$f_T(c)$ is finite.}

Let $x$ be an allocation of $v_1$ and define the characteristic
function $u_x$ as follows: $u_x(T)= f_T(x(T))$ for every coalition
$T$.

\begin{lemma} \label{lm:uniform-conver} For every $\varepsilon>0$  there
exists a time $m$ such that for every $c\le v_1(N)$, $T\subseteq N$
and an allocation $x$ of $v_1$ with $x(T)=c$,
$|u_x(T)-U^{t}_T(c)|<\varepsilon $ for every $t\ge m$.
\end{lemma}

\begin{proof} Fix $\varepsilon>0$ and a coalition $T$. Denote by $F$
the set of fixed points of $V(S;x)(T)$ in the interval $[0,
v_1(N)]$. Let $B$ be a finite subset of $F$ having the property that
for every $a\in F$ there is $b\in B$ such that $|a-b|<\varepsilon$. Since
$V(S;x)(T)$ is continuous and monotonically increasing in $x(T)$,
$f_T$ is monotonically increasing. Thus, $f_T(a)\in [b_1, b_2]$, for
every $a\in [b_1, b_2]$ with $b_1, b_2\in B$. Moreover, the distance
$|f_T(a)-U^{t}_T(a)|$ is decreasing with $t$. Denote by $A(b)$ the
set of all points that are absorbed to $b\in B$. That is,
$A(b)=\{a\colon f_T(a)=b\}.$

For every $b\in B$, there is time $m_b$ such that
$|b-U^{t}_T(a)|<\varepsilon $ for every $t\ge m_b$ and $a\in A(b)$. Let
$m_T=\max \{m_b \colon b\in B\}$. Thus, for every  $a\in [0,
v_1(N)]$, either $a\in A(b)$ for some $b\in B$, in which case
$|f_T(a)-U^{t}_T(a)|=|b-U^{t}_T(a)|<\varepsilon $ for every $t\ge m_T$, or
$|a-f_T(a)|\le |U^{t}_T(a)-f_T(a)|<\varepsilon$ for every $t$. Since there
are finitely many coalitions, $m=\max \{m_T \colon T\subseteq N\}$
satisfies the assertion of the lemma. \qed
\end{proof}

\begin{theorem} \label{th:stab} Consider a game where $V(S;x)(T)$
is continuously determined by  $x(T)$ in an increasingly monotonic
fashion. Assume furthermore, that $V(N;x)(N)=v_1(N)=1$ when
$x(N)=v_1(N)$. Then, the two following statements are equivalent:
\begin{enumerate}[{\rm (i)}]
\item\label{it:th:stab-1} For any $\varepsilon>0$ there is $\delta_0<1$ such
that for every $\delta\in [\delta_0, 1)$ the $\varepsilon$-stable
core of a game $(V,\delta)$ is not empty.

\item\label{it:th:stab-2} For every $\varepsilon>0$ there exists an allocation $x$ of
$v_1$ such that the  $\varepsilon$-core of $u_x$ is not empty.
\end{enumerate}
\end{theorem}

Before we proceed to the proof of this theorem, we need an auxiliary
result. Let $a_1,a_2,\dots $ be a bounded sequence of numbers. For
any integer $h$ denote, $a_*^{h,\delta}= (1-\delta)\sum_{t=h}^\infty
\delta^{t-h}a_{t}$. The proof of Theorem \ref{th:stab} uses the
following lemma.

\begin{lemma} \label{prop stab-core} For every $\delta<1$
large enough and every bounded sequence of numbers  $a_1,a_2,\dots $
such that $\{a_*^{h,\delta}\}_h$ has an accumulation point $a\ge 0$,
and every $\gamma>0$ there is a time $h$ such that $a_h> a-\gamma$,
while $a_*^h< a+\gamma$.
   \end{lemma}

\begin{proof} For every $h_1<h_2$,
\begin{equation}\label{eq sum}
a_*^{h_1,\delta}=(1-\delta)\sum_{t= h_1}^{h_2-1}\delta^{t-h_1}a_t
+(1-\delta)\left(1-\sum_{t= h_2}^{\infty}\delta^{t-h} \right)a_*^{h_2,\delta}.
\end{equation}
We assume that the sequence $a_1,a_2,\dots $ is bounded by $M\ge 1$.
Since $a$ is an accumulation point, there are $h_1<h_2$ such that
\[
\sum_{t= h_1}^{h_2-1}\delta^{t-h_1}>1- \frac{\gamma}{2M} \ \text{ and }
|a_*^{h,\delta}-a|< \frac{\gamma}{2},\ \text{ for }\ h=h_1,h_2.
\]

We consider the greatest $h$,  $h_1 \le h \le h_2$  such that
$a_h\ge a-\gamma$. There exists such $h$ because if all $h$ between
$h_1$ and $h_2$ satisfy
$a_h< a-\gamma$, then %by \eqref{eq sum}
\[
a_*^{h_1}< \left(1-\frac{\gamma}{2M})(a-\gamma\right)+\frac{\gamma}{2M}M<a-\frac{\gamma}{2},
\]
which contradicts the choice of $h_1$.

As for $a_*^h$, \eqref{eq sum} applied to $h=h_1$ and to $h_2$
implies that
\begin{align*}
a_*^{h,\delta}&=(1-\delta)\sum_{t= h}^{h_2-1}\delta^{t-h}a_t
+(1-\delta)\left(1-\sum_{t= h_2}^{\infty}\delta^{t-h} \right)a_*^{h_2,\delta}\\
&\le (1-\delta)a_h+(1-\delta)\sum_{t=
h+1}^{h_2-1}\delta^{t-h}(a-\gamma)+ (1-\delta)\left(1-\sum_{t=
h_2}^{\infty}\delta^{t-h} \right)(a+\gamma/2).
\end{align*}
When $\delta$ is large enough, $(1-\delta)a_h<\gamma/2$. Thus,
$a_*^{h, \delta}\le \gamma/2 +\delta (a+\gamma/2)\le a+\gamma,$ as
desired.  \qed
\end{proof}

\begin{proof} [Proof of Theorem \ref{th:stab}] We assume without
loss of generality that $v_1(N)=1$. We prove that
\eqref{it:th:stab-1} implies \eqref{it:th:stab-2}. We assume that
for any $\varepsilon>0$ the $\varepsilon$-stable core of
$(V,\delta)$ is not empty for $\delta$ sufficiently large. Fix
$\varepsilon>0$ and assume that $1-\varepsilon<\delta$.

Let $m$ be the one guaranteed by Lemma \ref{lm:uniform-conver} and
$\varepsilon$. Suppose that the discount factor $\delta$ is large enough so
the total payoff of any coalition during $m$ periods could not
exceed $\varepsilon$.

Let $x_1,x_2,\dots$ be in the $\varepsilon$-stable core of
$(V,\delta)$ and let $x_*$ be an accumulation point of the sequence
$x^h_*= x^h_*(\delta)=(1-\delta)\sum_{t=h}^\infty \delta^{t-h} x_t,~
h=1,2. \dots $ (it exists because, by assumption, the sequence of
allocations is bounded).

By assumption, $x_1(N)=x_2(N)=\dots$. We will show that $x_*$ is in
the $5\varepsilon$-core of $u_{x_*}$. If not, then there is a coalition $T$
such that $x_*(T)<u_{x_*}(T)-5\varepsilon.$ In particular, $x_*(T)$ is not
a fixed point of $f_T$. Since for every time $h$, $ x^h_*$ is an
average of $x_h,x_{h+1},\dots$, we have that $
x(N):=x^h_*(N)=V(N,x_t)(N)=v_1(N)$ (the last equality is by
assumption)
for every period $t$, implying that $x_*$ is
an allocation of $v_1$.

Since the set of fixed points of $f_T$ is closed, and $x_*(T)$ is not
a fixed point of $f_T$, one can find $\beta>0$ such that

 \begin{equation}\label{eq
beta-eps}x^h_*(T)>x_*(T)-\beta \ \text{ implies }  \
    f_T(x^h_*(T))\ge f_T(x_*(T)).
\end{equation}

The reason is that when $x^h_*(T)$ is not a fixed point of $f_T$,
there is an open interval around $x^h_*(T)$, whose points
have all the same range as $x^h_*(T)$. That is, $f_T$ is fixed around
$x^h_*(T)$. Thus, there is $\beta>0$ such that
$x^h_*(T)>x_*(T)-\beta$ implies $ f_T(x^h_*(T))= f_T(x_*(T))$. In
case $x^h_*(T)>x_*(T)$, then by monotonicity $f_T(x^h_*(T))\ge
f_T(x_*(T))$ and therefore, (\ref{eq
beta-eps}).

Applying Lemma~\ref{prop stab-core} to the sequence $x_1(T), x_2(T),
\dots$  and the accumulation point $x_*(T)$, we obtain that for
$\gamma=\min(\beta,\varepsilon)$, there is a time $h$ such that
\begin{equation}\label{eq 3}
 x_h(T)> x_*(T)-\gamma \ \text{ and } \  x^h_*(T)< x_*(T)+\gamma.
\end{equation}
 In words, the instantaneous  payoffs of coalition $T$ at time $h$ is greater than  $x_*(T)-\gamma$, while the present value of
coalition $T$'s payoff at time $h+1$ is less than  $x_*(T)+\gamma$.

We now describe a deviation of coalition $T$. At time $h+1$
coalition $T$ deviates and plays the game $V(T,x_h(T))$. From time
$h+2$ on, any allocation of the stage-game is fine. After $m$
periods the worth of the coalition $T$ in the stage-game, is by
Lemma \ref{lm:uniform-conver}, close to $u_{x_h}(T)$ up to $\varepsilon$.
Since the first $m$ periods after $h$ contribute at most $\varepsilon$ to
the entire present value of coalition $T$, the payoff (discounted to
time $h+1$) for $T$ due to the deviation is at least
$u_{x_h}(T)-2\varepsilon$.

By (\ref{eq 3}), $x_h(T)> x_*(T)-\gamma \ge x_*(T)-\beta$. Because of (\ref{eq beta-eps}),  $u_{ x_h}(T)\ge u_{
x_*}(T)$. Thus, due to the deviation, the payoff of coalition $T$ is
at least $u_{x_*}(T)-2\varepsilon$. However, the sequence
$x_1,x_2,\dots$ is in the $\varepsilon$-stable core of
$(V,\delta)$, and therefore, by deviating coalition $T$ cannot get more than $\varepsilon$ beyond the originally planned payoff  ($x^{h+1}_*(T)$). Thus,
\begin{equation}\label{eq 1}
    u_{x_*}(T)-2\varepsilon\le x^{h+1}_*(T)+\varepsilon.
\end{equation}

Again from (\ref{eq 3}) we have $x^h_*(T)< x_*(T)+\gamma \le x_*(T)+\varepsilon$. Thus,
\begin{equation}\label{eq 2}
x^{h+1}_*(T) \le x^h_*(T)+\varepsilon \le x_*(T)+2\varepsilon.
\end{equation}

The first inequality of (\ref{eq 2}) is due to the fact that $1-\delta<\varepsilon$ (as assumed at the beginning of the proof) and
$x_h(T)\le v_1(N)=1$. Hence, by (\ref{eq 1})
and (\ref{eq 2}), $u_{x_*}(T)-2\varepsilon\le x^{h+1}_*(T)+\varepsilon\le
x_*(T)+3\varepsilon.$ It implies that $u_{x_*}(T)-5\varepsilon \le x_*(T).$ This
is a contradiction, and therefore $x_*$ is the $5\varepsilon$-core of
$u_{x_*}$.

The proof that \eqref{it:th:stab-2} implies \eqref{it:th:stab-1} is relatively easy and is therefore
omitted.  \qed
\end{proof}

\begin{remark} Assertion \eqref{it:th:stab-2} of Theorem \ref{th:stab} states that
for every $\varepsilon>0$ there exists an allocation $x$ of $v_1$ such that
the  $\varepsilon$-core of $u_x$ is not empty. Due to lack of continuity it
is impossible to conclude that there exists an allocation $x$ of
$v_1$ such that the core of $u_x$ is not empty. This is so because
when a sequence of allocations of $v_1$, say $(x_k)_k$ (each in the
$\varepsilon$-core of the respective $u_{x_k}$), is converging to $x$,
there is no guarantee that $u_{x_k}$ converges to $u_x$.
\end{remark}

\section{The credible core}\label{sec def-core-credible}

The third type of core we are about to define is close in spirit to
subgame perfect equilibrium in the theory of non-cooperative games.
A dissatisfied coalition may deviate at any time in which future
allocations guarantee less than it can do alone. Hence, stability
conditions must be preserved not only at the beginning of the game,
but throughout the entire game. But when creating its own game, a
coalition, say $S$ may face a threat from one of its sub-coalitions,
say $T$. The game established by  $T$ may depend on the entire history
of allocations, starting from the grand coalition allocation at the
beginning of game, continuing with $S$-allocation and ending with
allocations of its own.

The game may start with the grand coalition, run this way for a
while and only then coalition $S_1$ may decide to deviate, run for a
while and then coalition $S_2$ may deviate etc. Thus, histories now
consist of $x^{S_1}_1,x^{S_2}_2,\dots,x^{S_t}_t $, where $S_1 =N$
and the sequence of $S_{\ell}$ is decreasing (w.r.t. inclusion). A
history $x^{S_1}_1,x^{S_2}_2,\dots,x^{S_t}_t $ is feasible if at any
time $t$, $x^{S_t}_t$ is an allocation of the stage game
$V^{S_t}(x^{S_{t-1}}_{t-1}). $

A central planner has an allocation policy, denoted $\sigma$, that
associates with any time and any possible instantaneous game an
allocation. In other words, the central planner has a full
contingency plan as to how the available cake should be split at any
point in time. Formally, the \emph{allocation policy} $\sigma$ is
such that $\sigma(x^{S_t}_t)$ is an
allocation of the game
$V(S_{t+1};x^{S_t}_t)$.

After any history of allocations,
$x^{S_1}_1,x^{S_2}_2,\dots,x^{S_t}_t$, an {allocation policy}
$\sigma$ determines uniquely a continuation stream of allocations:
$x^{S_t}_{t+1}=\sigma(x^{S_t}_t)$,
$x^{S_t}_{t+2}=\sigma(x^{S_t}_{t+1})$,
etc. We denote this continuation by
$C_{\sigma}(x^{S_1}_1,x^{S_2}_2,\dots,x^{S_t}_t)$.

\begin{definition}[Credible core]
An allocation policy $\sigma$ is in the \emph{credible core with
discount factor $\delta$} of $(V,\delta)$ if for every history
$x^{S_1}_1,x^{S_2}_2,\dots,x^{S_t}_t$ the sequence
$C_{\sigma}(x^{S_1}_1,x^{S_2}_2,\dots,x^{S_t}_t)$ of allocations is
in the stable core of
$V(S_{t+1};x^{S_t}_t)$.
\end{definition}

A similar idea has been used in a non-cooperative context by  \citet{BerPelWhi:JET1987, BerWhi:JET1987}. These authors consider Nash equilibria that are immune from deviations not just of single players but also of coalitions. Not every deviation is acceptable, though: a deviation of some coalition has to be in turn immune from deviations of sub-coalitions.

The main result of this section shows that a form of the \emph{one-deviation principle} holds for the credible core. This principle goes back to \citet{Bla:AMS1965} and has been widely used in extensive form noncooperative games. Recently  \citet{Var:ACE352008} applied it in the study of coalition formation in a cooperative context.

\begin{definition} A coalition $S\subsetneqq S_t$ has a \emph{profitable
one-deviation} after the history
$x^{S_1}_1,x^{S_2}_2,\dots,x^{S_t}_t$ from $\sigma$, if there is an
allocation $ x^{S}_{t+1}$ of
$V(S;x^{S_t}_t)$ such that the present
value for coalition $S$ of the sequence that starts at $
x^{S}_{t+1}$ and continues with $
C_{\sigma}(x^{S_1}_1,x^{S_2}_2,\dots,x^{S_t}_t, x^{S}_{t+1})$ is
greater than the present value for $S$ of the planned sequence
$C_{\sigma}(x^{S_1}_1,x^{S_2}_2,\dots,x^{S_t}_t)$.
\end{definition}

\begin{theorem}[The one-deviation principle]\label{th:one-deviation} An allocation policy $\sigma$ is
in the credible core if and only if after every history
$x^{S_1}_1,x^{S_2}_2,\dots,x^{S_t}_t$ and for every $S\subsetneq
S_t$ there is no profitable one-deviation.
\end{theorem}

\begin{proof} The `only if' direction is trivial. For the `if' part, assume
that after every history $x^{S_1}_1,x^{S_2}_2,\dots,x^{S_t}_t$ and
for every $S\subsetneq S_t$ there is no profitable one-deviation.
Suppose that there is a coalition $S\subsetneq S_t$ that has a
profitable deviation after the sequence
$x^{S_1}_1,x^{S_2}_2,\dots,x^{S_t}_t$. Denote the gain by $a>0$. By
continuity we may assume that this deviation consists of finite
number stages. Consider the shortest deviation of $S$ after
$x^{S_1}_1,x^{S_2}_2,\dots,x^{S_t}_t$ that guarantees a gain of at
least $a$. This deviation is $x^{S}_{t+1},\dots,x^{S}_{t+\ell}$. It
implies that the deviation $x^{S}_{t+1},\dots,x^{S}_{t+\ell-1}$
guarantees $S$ less than $a$, meaning that the single deviation of
$S$ to $x^{S}_{t+\ell}$ after the entire history
$x^{S_1}_1,x^{S_2}_2,\dots,x^{S_t}_t, x^{S}_{t+1},\dots
,x^{S}_{t+\ell-1}$ makes a positive gain for $S$. This shows that
there is no profitable one-deviation, as this direction of the
theorem claims.   \qed
\end{proof}

Note that the theorem does not say that when the allocation policy
$\sigma$ is in the credible core, the instantaneous allocations are
in the respective stage game.

\begin{remark}
Theorem~\ref{th:one-deviation} holds also under a richer non-Markovian dynamic, where $V$ depends on the entire history of allocations and not only on the last allocation.
\end{remark}

\section{Final remarks}

We close the paper with some additional comments.

\subsubsection*{More on the non-emptyness of the core}
Consider a game where $V(N;\cdot)(N)$ is constant. In addition
assume that for every coalition $S$, $V(N;\cdot)(S)$ is bounded. The
$\varepsilon$-stable core with discount factor $\delta$ is not empty if
and only if the core of $v_{*}(\cdot,\delta) $ is non-empty.

\subsubsection*{No-short assumption} Throughout the paper we assumed
that the stage allocation, $x^S_t$, satisfies two assumptions.
First, the allocation of player $i$ is at least her $v(i)$, that is,
$x^S_t(i)\ge B$, and second, $x^S_t$ is locally efficient, that is,
$x^S_t(S)=v^S_t(S)$. This assumption assumes that the inter-temporal
transfers are limited. That is $B$ might be well below $v_t(i)$ at a
certain moment, but since the overall payoff in the entire dynamic
game needs to be individually rational, the instantaneous payoffs
need to be sometimes higher than the stage IR level. The technical
advantage of these assumptions is that the set of possible
allocations at any stage is compact.

\subsubsection*{Random games} We analyze games where the instantaneous
game depends deterministically on the history. The issue of
stochastic dynamic game where the stage games are endogenously
determined remains open for further studies.

\subsection*{Acknowledgements}
The authors thank Sandro Brusco for helpful comments and relevant references.

\bibliographystyle{artbibst}  
\bibliography{bibdynamic}

\end{document}